\documentclass[a4paper]{article}
\usepackage{authblk}
\usepackage{lmodern}
\usepackage{indentfirst}
\usepackage{scrextend}
\usepackage{amssymb,amsthm,amsmath}
\DeclareMathOperator{\lcm}{lcm}
\DeclareMathOperator{\ord}{ord}
\usepackage{cmap}
\usepackage[TS1,T1]{fontenc}
\usepackage{lmodern}
\usepackage{textcomp}
\usepackage[utf8]{inputenc}
\usepackage{amsmath}
\usepackage{hyperref}
\usepackage{graphicx}
\usepackage{tikz}
\usepackage{tikz-cd}
\usetikzlibrary{automata,arrows,
positioning,shapes,calc}
\usepackage{epigraph}
\theoremstyle{plainsl}
  \newtheorem{theorem}{Theorem}[section]
  \newtheorem{proposition}[theorem]{Proposition}
  \newtheorem{lemma}[theorem]{Lemma}

\theoremstyle{definition}

 \theoremstyle{remark}

\def\calM{\mathcal{M}}
\def\Z{\mathbb{Z}}
\newcommand{\bigslant}[2]{{\raisebox{.2em}{$#1$}\left/\raisebox{-.2em}{$#2$}\right.}}

\usepackage{ulem,xcolor}

\begingroup
\def\2#1{\ifnum#1<10 0\fi\the#1}
\xdef\isodayandtime{
  {\2\day-\2\month-\the\year\space\2{\count0}:%
  \2{\count2}}}
\endgroup

\begin{document}
\title{\textbf{Elementary remarks about\\ 
Pisano periods.}}
\author[1]{G. H. E. Duchamp\thanks{gheduchamp@gmail.com}} 
\author[2]{P. Simonnet\thanks{simonnet\_p@univ-corse.fr}}
\affil[1]{\small{LIPN, Northen Paris University, Sorbonne Paris City,\break 93430 Villetaneuse, France.}}
\affil[2]{\small{Université de Corse,\break 20250 Corte, France.}}
\date{\isodayandtime}
\maketitle
\begin{abstract}
In this short note, we reprove in a very elementary way some known facts about Pisano periods as well as some considerations about the link between Pisano periods and the order of roots of the characteristic equation. The technics only requires a small background in ring theory (merely the definition of a commutative ring). The tools set here can be reused for all linear recurrences with quadratic non-constant characteristic equation.        
\end{abstract}

\subsection*{Keywords: \textmd{Pisano periods, ring extensions, order.}}

\section{Introduction}

Let $R$ be a commutative ring. The Fibonacci sequence within $R$ is defined by the following recurrence. 
\begin{equation}
F_0=0_R,\ F_1=1_R,\ F_{n+2}=F_n+F_{n+1}
\end{equation} 
It can be easily checked that this sequence is periodic iff the characteristic of the ring $R$ is positive\footnote{Although all the Fibonacci sequence lives in the subring $\Z.1_R$ ($\simeq \bigslant{\Z}{n\Z}$ where $n$ is the characteritic of $R$), it be more compact (and a bit more elegant, due to the possible adjuntions)   
to consider the sequence as living in an arbitrary ring $R$. Some of the tools set here can be reused to cope with linear recurrences with quadratic non-constant (and invertible) characteristic equation.}.\\ 
The period of $F_n$ (called Pisano period) is exactly the order of the invertible matrix
\begin{equation}
\begin{pmatrix}
0 & 1\\
1 & 1
\end{pmatrix}
\in \calM(2,R)
\end{equation} 
For introducing matters and results with respect to Pisano periods, see \cite{Pisano}.\\
The aim of this short note is to write down some elementary facts about these periods and how, according 
to cases, it is exactly related to orders of some units in $R$.   

\section{Setting and results.} 

We discuss around the possiblity of solving the equation $X^2-X-1=0$ within $R$.

Throughout the paper, we suppose that $1_R\not=0_R$ (i.e. $R$ is not the null ring) and that $5=5.1_R$ is a unit in $R$. We will call (E) the preceding equation i.e.
\begin{equation*}
\hspace*{3cm} x^2-x-1=0\hspace*{3cm} \mbox{(E)}
\end{equation*}
and make use of $\sigma$ the automorphism of $R[x]$  sending $x\mapsto 1-x$.  

\subsubsection*{First case: $x^2-x-1=0$ has (at least) one root}

We have the following lemma

\begin{lemma}\label{lem1} We suppose that the equation (E)  has at least one root in $R$ (call it $r$ and set 
$s:=1-r$), then 
\begin{enumerate}
\item $s$ is a root of (E)
\label{item1}
\item $r$ and $s$ are units in $R$
\label{item2}
\item $r-s$ is a square root of $5$ and therefore also a unit in $R$ 
\label{item3}
\item $F_n=\frac{r^n-s^n}{r-s}$ (called Binet formula in (\cite{Pisano}, see also Max Alekseyev's answer in \cite{MO1}).
\label{item4}
\item The following matrix 
\begin{equation}
\begin{pmatrix}
1 & 1\\
r & s
\end{pmatrix}
\end{equation} 
is a unit in $\calM(2,R)$ and 
\label{item5}
\item We have in $\calM(2,R)$
\begin{equation}\label{conj1}
\begin{pmatrix}
1 & 1\\
r & s
\end{pmatrix}^{-1}
\begin{pmatrix}
0 & 1\\
1 & 1
\end{pmatrix}
\begin{pmatrix}
1 & 1\\
r & s
\end{pmatrix}
=
\begin{pmatrix}
r & 0\\
0 & s
\end{pmatrix}
\end{equation}
\label{item6}
\item Then the period of $F_n$ within $R$ is 
infinite iff $r$ or $s$ is of infinite order. 
\label{item7}
\end{enumerate} 
\end{lemma}

\begin{proof}
\ref{item1}) Remarking that $\sigma(x^2-x-1)=x^2-x-1$ we get the result.\\
\ref{item2}) $r(-s)=r(r-1)=r^2-r=1$ then $r^{-1}=-s$ and similarly $s^{-1}=-r$\\
\ref{item3}) $$(r-s)^2=r^2-2rs+s^2\underbrace{=}_{rs=-1}=r^2+s^2+2\underbrace{=}_{\substack{r^2=r+1\\ 
s^2=s+1\\ r+s=1}}=5$$ 
Now, if $y^2=5$ and as $5$ is a unit, we have 
$y.(y/5)=y^2.5=1$ which proves that $\pm(r-s)$ are units in $R$.\\
\ref{item4}) $(r^n)$ and $(s^n)$ are sequences satisfying 
$x_{n+2}=x_n+x_{n+1}$ and the same holds for all their linear combinations in particular 
$$
t_n=\dfrac{r^n-s^n}{r-s}
$$  
it suffices now to check that $t_i=F_i$ for $i=0,1$. QED 
\\
\ref{item5}) One has 
\begin{equation}
\begin{vmatrix}
1 & 1\\
r & s
\end{vmatrix}
=s-r
\end{equation}
which is a unit from point \eqref{item3}. Therefore 
\begin{equation}
\begin{pmatrix}
1 & 1\\
r & s
\end{pmatrix}^{-1}
=\dfrac{1}{s-r}
\begin{pmatrix}
s & -1\\
-r & 1
\end{pmatrix}
\end{equation}
\\
\ref{item6}) We have 
\begin{equation}
\begin{pmatrix}
0 & 1\\
1 & 1
\end{pmatrix}
\begin{pmatrix}
1 & 1\\
r & s
\end{pmatrix}
=
\begin{pmatrix}
1 & 1\\
r & s
\end{pmatrix}
\begin{pmatrix}
r & 0\\
0 & s
\end{pmatrix}
\end{equation}
Hence \eqref{conj1}.
\\
\ref{item7}) Again, from equation \eqref{conj1}, one 
has in 
$\calM(2,R)^{\times}$ and $R^\times$
\begin{equation}
\ord(\begin{pmatrix} 0 & 1\\1 & 1\end{pmatrix})=
\lcm(\ord(r),\ord(s))
\end{equation} 
\end{proof}

We are now in the position to state

\begin{proposition} We suppose the hypotheses of Lemma \ref{lem1}. Then 
\begin{enumerate}
\item The roots, $r,s$ of (E) are both of finite or infinite order.\label{item8}
\item In the case when both roots are of finite order
\begin{enumerate}
\item Either $2.1_R=0_R$ and $\ord(r)=\ord(s)=3$.\label{item8.5}
\item Or ($2\not=0$ and) one of the orders of $r,s$ is odd (say $\ord(r)\equiv\ 1\mod 2$) and the other one is the double of it ($\ord(s)=2\ord(r)$).\label{item9}
In this case, the Pisano period is $\ord(s)=2\ord(r)$
\item Or ($2\not=0$ and) their orders are all even.\label{item10}
In this case $\ord(s)=\ord(r)$ and the Pisano period is 
$\ord(s)=\ord(r)$.
\end{enumerate}
\end{enumerate}
\end{proposition} 
\begin{proof}
\ref{item8}) Due to the fact that 
\begin{equation}\label{invol1}
s=(-1/r) \mbox{ and }r=(-1/s)
\end{equation}
\item \ref{item8.5}) For $g\in\{r,s\}$, the group generated by $g$ is $\{1,g,g+1\}$ hence $g$ is of order $3$. From now on we assume that $2\not=0$.\\
\ref{item9}) Let us now suppose that 
$\ord(r)\equiv\ 1\mod 2$. Then, due to \eqref{invol1}, we have 
\begin{equation}\label{cruc1}
1=\left(s^{\ord(s)}\right)^{\ord(r)}=
\left((-1/r)^{\ord(r)}\right)^{\ord(s)}=(-1)^{\ord(s)}
\end{equation}
which, due to the fact that $2.1_R\not= 0_R$, entails that $s$ is of even order.\\
But, in the same way, equation \eqref{invol1} shows that 
$s^{2.\ord(r)}=(-1)^{2.\ord(r)}=1$ and then 
$\ord(s)|2.\ord(r)$.\\
On the other hand 
$$
r^{\ord(s)}=(-1/s)^{\ord(s)}=(-1)^{\ord(s)}=1
$$
this proves that $\ord(r)|\ord(s)$. By Gauss lemma and taking into account that $\ord(r)$ is odd and $\ord(s)$ is even, we also have $2.\ord(r)|\ord(s)$.\\
Finally $\ord(s)|2.\ord(r)|\ord(s)$ therefore  
$\ord(s)=2.\ord(r)$.\\
\ref{item10}) Similar reasoning using \eqref{invol1}.   
\end{proof} 

\subsubsection*{Second case: $x^2-x-1=0$ has no roots in 
$R$.} 

We consider $R_1:=\bigslant{R[x]}{(x^2-x-1)}$ and, with 
$r=\bar{x}$ (considering that $5$ is still a unit in $R_1$), we are brought back to the preceding case. In this case, however, we remark that we have the automorphism 
$\sigma$ sending $x$ to $(1-x)$ which preserves the ideal 
$(x^2-x-1)$ and then passes to quotients (see Figure \ref{Fig1}). 
\begin{figure}[h!]
\centering
\begin{tikzcd}[column sep=2cm, row sep=0.8cm]
R[x] \arrow[r,"\sigma"]
\arrow[d, "p"]&  
R[x]\arrow[d, "p"]\\
\bigslant{R[x]}{(x^2-x-1)}\arrow[r,"\bar{\sigma}"] & 
\bigslant{R[x]}{(x^2-x-1)}
\end{tikzcd}
\caption{Involutive automorphism of 
$\bigslant{R[x]}{(x^2-x-1)}$ without Galois theory.}\label{Fig1}
\end{figure}
this proves that $r$ and $\bar{\sigma}(r)=1-r$ are of the same order\footnote{This common order is necessarily even by the examination of the first case. }.

\end{document}